\documentclass[manuscript]{acmart}
\usepackage{etex}
\usepackage{float}
\usepackage{booktabs}
\usepackage{boxedminipage}
\usepackage{epsfig}  
\usepackage{multirow} 
\usepackage{url}
\usepackage{schemepgm}
\usepackage{graphicx}
\usepackage{graphpap}
\usepackage{tabularx}
\usepackage{bcprules}
\usepackage{subfig}
\usepackage[lined, boxed]{algorithm2e}
\usepackage{makecell}
\usepackage{balance}
\usepackage{tikz}
\usepackage{dsfont}                                  

\usepackage{enumitem}
\setlist{nolistsep,leftmargin=*}
\usepackage{xifthen}
\usepackage[normalem]{ulem}
\usepackage{amsmath}

\definecolor{Myblue}{rgb}{.2,0,1}
\definecolor{Myred}{rgb}{1,.2,.2}
\newcommand{\cmt}[1] {}

\def\myvec{\mathaccent"017E } 
\newtheorem{property}{Property}

\newcommand{\simpl}[1]{\ensuremath{\mathcal{S}(#1)}}
\newcommand{\canon}{\ensuremath{\mathcal{C}}}
\newcommand{\csimpl}{\ensuremath{\stackrel{\canon}{\rightarrow}}}
\newcommand{\ssimpl}{\ensuremath{\stackrel{\Simpl}{\rightarrow}}}

\newcommand{\acar}{\ensuremath{\mathbf{0}}}
\newcommand{\acdr}{\ensuremath{\mathbf{1}}}
\newcommand{\bcar}{\ensuremath{\bar\acar}}
\newcommand{\bcdr}{\ensuremath{\bar\acdr}}
\newcommand{\clazy}{\ensuremath{{\mathbf{2}}}}

\newcommand{\epsilonset}{\ensuremath{\lbrace\epsilon\rbrace}}
\newcommand{\var}[1]{\raisebox{.4mm}{\ensuremath{#1}}}

\newcommand{\Lfonly}{\ensuremath{\mathsf{DS}}}
\newcommand{\Lf}[3]{\ensuremath{\DT_{\mathit #1}^{#2}( {\mathit #3})}}
\newcommand{\Lfun}[2]{\ensuremath{\Lfonly_{\mathit #1}^{#2}}}
\newcommand{\Demand}[1]{\ensuremath{\mathsf{D}_{#1}}}
\newcommand{\Demandp}[1]{\ensuremath{\mathsf{D}'_{#1}}}
\newcommand{\lang}[1]{\ensuremath{\mathsf{L}({#1})}}

\newcommand{\mainpgm}{\ensuremath{\mathbf{main}}}

\newenvironment{proofoutline}[1][Proof Outline]{\begin{trivlist}
  \item[\hskip \labelsep {\sc #1.}]}{\end{trivlist}}

\newcommand{\len}{\ \vdash^l\ }

\newcommand{\pp}[2]{\ensuremath{#2}} 
\newcommand{\removed}{\ensuremath{\square}}

\newcommand{\stack}    {\ensuremath{\mathbf{S}}}
\newcommand{\heap}     {\ensuremath{\mathbf{H}}}
\newcommand{\demand}   {\ensuremath{\sigma}}
\newcommand{\DT}       {\ensuremath{\mathds{DS}}}

\newcommand{\A}      {\ensuremath{\mathcal{A}}}

\newcommand{\mIF}     {\ensuremath{\mathbf{if}}}
\newcommand{\mLET}    {\ensuremath{\mathbf{let}}}
\newcommand{\mNIL}    {\ensuremath{\mathbf{nil}}}
\newcommand{\mCAR}    {\ensuremath{\mathbf{car}}}
\newcommand{\mCDR}    {\ensuremath{\mathbf{cdr}}}
\newcommand{\mCONS}   {\ensuremath{\mathbf{cons}}}
\newcommand{\mRETURN} {\ensuremath{\mathbf{return}}}
\newcommand{\mNULLQ}  {\ensuremath{\mathbf{null?}}}
\newcommand{\mIN}     {\ensuremath{\mathbf{in}}}

\newcommand{\infscale}[1]{\scalebox{.99}{#1}}
\newcommand{\sinfrule}[3]{\ifthenelse{\isempty{#1}}{
  \infrule{\infscale{\ensuremath{#2}}}
          {\infscale{\ensuremath{#3}}}}{%
  \infrule[\infscale{#1}]
          {\infscale{\ensuremath{#2}}}
          {\infscale{\ensuremath{#3}}}}}

\newcommand{\demandlang}{\ensuremath{(\acar+\acdr)^\ast}}
 
\definecolor{myblue}{rgb}{0.0,0.0,0.5}

\begin{document}

\setlength{\pdfpageheight}{\paperheight}
\setlength{\pdfpagewidth}{\paperwidth}

\title{An Incremental Slicing Method for Functional Programs}

\author{Prasanna Kumar K.}
\affiliation{
  \department{Department of Computer Science}
  \institution{IIT Bombay}         
  \city{Mumbai}
  \postcode{400076}
  \country{India}
}
\email{prasannak@cse.iitb.ac.in}   

\author{Amitabha Sanyal}
\affiliation{
  \department{Department of Computer Science}
  \institution{IIT Bombay}         
  \city{Mumbai}
  \postcode{400076}
  \country{India}
}
\email{as@cse.iitb.ac.in}          

\author{Amey Karkare}
\affiliation{
  \department{Department of Computer Science} 
  \institution{IIT Kanpur}         
  \city{Kanpur}
  \postcode{208016}
  \country{India}
}

\email{karkare@cse.iitk.ac.in}     

\begin{abstract}
  Several applications of slicing require  a program to be sliced with
  respect to more than one slicing criterion.  Program specialization,
  parallelization  and  cohesion  measurement  are  examples  of  such
  applications.  These  applications can  benefit from  an incremental
  static  slicing  method  in  which   a  significant  extent  of  the
  computations  for slicing  with respect  to one  criterion could  be
  reused  for another.   In this  paper,  we consider  the problem  of
  incremental slicing of functional programs.
    
  We first present a non-incremental  version of the slicing algorithm
  which   does  a   polyvariant  analysis\footnote{In   a  polyvariant
    analysis~\cite{Smith:2000}, the  definition of a function  in some
    form  is  re-analyzed multiple  times  with  respect to  different
    application contexts.}   of functions. Since  polyvariant analyses
  tend to be costly, we  compute a compact context-independent summary
  of each function and then use this  summary at the call sites of the
  function. The  construction of  the function summary  is non-trivial
  and  helps  in the  development  of  the incremental  version.   The
  incremental  method  on  the  other  hand  consists  of  a  one-time
  pre-computation step that uses  the non-incremental version to slice
  the program  with respect to  a fixed default slicing  criterion and
  processes the results  further to a canonical  form.  Presented with
  an  actual  slicing  criterion,  the  incremental  step  involves  a
  low-cost computation that uses the results of the pre-computation to
  obtain the slice.

  We have implemented a prototype of the  slicer for a pure subset of
  Scheme,  with  pairs and lists  as the only algebraic  data types.
  Our experiments show that the incremental step of the slicer
  runs  orders  of  magnitude  faster  than  the  non-incremental
  version.    We have also proved the
  correctness of  our incremental  algorithm with respect  to the
  non-incremental version.

\end{abstract}
\acmYear{2017}
\acmMonth{9}

\maketitle

\section{Introduction}
\label{sec:intro}
Program slicing refers to the class of techniques that delete parts of
a  given  program while  preserving  certain  desired behaviors,  for
example, memory, state or parts of output. These behaviors are called
{\em  slicing criteria}.   Applications  of slicing  include
debugging     (root-cause    analysis),     program    specialization,
parallelization  and cohesion  measurement.  However,  in some  of the
above applications,  a program has to  be sliced more than  once, each
time  with a  different  slicing criterion.  In  such situations,  the
existing
techniques~\cite{weiser84, horwitz88, reps96, Liu:2003,
  Rodrigues_jucs_12_7, Silva_System_Dependence_Graph}   are
inefficient as they typically analyze the program multiple times.
Each round of analysis involves  a fixed point computation on the
program text or some intermediate  form of the program, typically
SDG  in the  case of  imperative languages.   We thus  require an
incremental  approach to  slicing which can avoid  repeated  fixpoint
computation  by reusing  some of  the information  obtained while
slicing the same program earlier with a different criterion.  


\begin{figure}[t!]
  \centering
  \renewcommand{\arraystretch}{1.1}
  \begin{tabular}{@{}|@{\ }p{.31\textwidth}|p{.31\textwidth}|p{.31\textwidth}|@{}} \hline
    \begin{uprogram}
      \UFL~~(\DEFINE\ (\linecharcount~\str~\lc~\cc)
      \UNL{1} (\SIF\ (\NULLQ~\str)
      \UNL{3}        (\SRETURN~(\CONS\ \lc\ \cc))
      \UNL{3}        (\SIF~(\EQ~(\CAR~\str)~$nl$)
      \UNL{5} (\SRETURN\ (\linecharcount~(\CDR\ \str)
      \UNL{14} (+ \lc\ 1)
      \UNL{14} (+ \cc\ 1)))
      \UNL{5} (\SRETURN\ (\linecharcount~(\CDR\ \str)
      \UNL{13} $\pi_1$:\lc
      \UNL{13} $\pi_2$:(+ \cc\ 1))))))
      \UNL{0}
      \UNL{0} $(\DEFINE\ (\mainpgm)$
      \UNL{1} $(\SRETURN\ (\linecharcount\ \ldots\ 0~0))))$
    \end{uprogram} &
    \begin{uprogram}
      \UFL(\DEFINE\ (\linecharcount~\str~\lc~\removed)
      \UNL{1} (\SIF\ (\NULLQ~\str)
      \UNL{3}        (\SRETURN~(\CONS\ \lc\ \removed))
      \UNL{3}        (\SIF~(\EQ~(\CAR~\str)~$nl$)
      \UNL{5} (\SRETURN\ (\linecharcount~(\CDR\ \str)
      \UNL{14} (+ \lc\ 1)
      \UNL{14} ~\removed))
      \UNL{5} (\SRETURN\ (\linecharcount~(\CDR\ \str)
      \UNL{13} $\pi_1$:\lc
      \UNL{13} $\pi_2$:~\removed)))))
      \UNL{0}
      \UNL{0} $(\DEFINE\ (\mainpgm)$
      \UNL{1} $(\SRETURN\ (\linecharcount\ \ldots\ 0~\removed))))$
    \end{uprogram} &
    \begin{uprogram}
      \UFL(\DEFINE\ (\linecharcount~\str~\removed~\cc)
      \UNL{1} (\SIF\ (\NULLQ~\str)
      \UNL{3}        (\SRETURN~(\CONS\ \removed\ \cc))
      \UNL{3}        (\SIF~(\EQ~(\CAR~\str)~$nl$)
      \UNL{5} (\SRETURN\ (\linecharcount~(\CDR\ \str)
      \UNL{14} \removed
      \UNL{14} (+ \cc\ 1)))
      \UNL{5} (\SRETURN\ (\linecharcount~(\CDR\ \str)
      \UNL{13} $\pi_1$:~\removed
      \UNL{13} $\pi_2$:~(+ \cc\ 1))))))
      \UNL{0}
      \UNL{0} $(\DEFINE\ (\mainpgm)$
      \UNL{1} $(\SRETURN\ (\linecharcount\ \ldots\ \removed~0))))$
    \end{uprogram} \\ \hline
    (a) Program to compute the number of lines and characters in a string. &
    (b) Slice of program in (a) to compute the number of lines only. &
    (c) Slice of program in (a) to compute the number of characters only.
    \\ \hline
  \end{tabular}
    \caption{A program in Scheme-like language and its slices. The parts that are sliced away are denoted by \removed.}\label{fig:mot-example}
\end{figure}

The        example       from~\cite{reps96}        shown       in
Figure~\ref{fig:mot-example}b motivates the  need for incremental
slicing. It shows a simple  program in a Scheme-like language. It
takes a  string as  input and  returns a  pair consisting  of the
number    of    characters    and   lines    in    the    string.
Figure~\ref{fig:mot-example}b shows the program when it is sliced
with respect  to the first  component of the output  pair, namely
the number of lines in the  string ({\tt lc}).  All references to
the  count   of  characters   ({\tt  cc})  and   the  expressions
responsible for  computing {\tt  cc} only  have been  sliced away
(denoted $\removed$).   The same  program can  also be  sliced to
produce only the char count and the resulting program is shown in
Figure~\ref{fig:mot-example}c.  \mbox{}

The  example illustrates  several important  aspects for  an effective
slicing  procedure.  We  need the  ability to  specify a  rich set  of
slicing  criteria to  select  different parts  of  a possibly  complex
output structure (first and second component of the output pair in the
example, or say,  every even element in an output  list).  Also notice
that to compute some part of  an output structure, all prefixes of the
structure have to be computed.  Thus, slicing criteria have to be {\em
  prefix-closed}.   Finally, it  seems likely  from the  example, that
certain  parts  of   the  program  will  be  present   in  any  slice,
irrespective  of the  specific slicing  criterion\footnote{the trivial
  {\em null} slicing  criteria where the whole program  is sliced away
  is  an  exception, but  can  be  treated separately.}.   Thus,  when
multiple slices of the same  program are required, a slicing procedure
should strive for efficiency  by minimizing re-computations related to
the common parts.

In  this paper,  we consider  the problem  of incremental  slicing for
functional programs. We restrict ourselves  to tuples and lists as the
only  algebraic data  types.   We represent  our  slicing criteria  as
regular grammars that  represent sets of prefix-closed  strings of the
selectors \CAR\/ and \CDR.  The  slicing criterion represents the part
of the output of  the program in which we are  interested, and we view
it  as being  a  \emph{demand} on  the program.   We  first present  a
non-incremental  slicing  method, which propagates the demand represented
by  the slicing  criterion  into  the program.
In   this    our   method   resembles the   projection    function   based
methods of~\cite{reps96, Liu:2003}.  However, unlike these methods, we do
a context-sensitive analysis of functions calls. This makes our method
precise by  avoiding analysis  over infeasible  interprocedural paths.
To  avoid the  inefficiency  of  analyzing a  function  once for  each
calling context,  we create a compact  context-independent summary for
each function. This summary is then  used to step over function calls.
As we shall see, it is this context independent summary that also makes the
incremental version possible in our approach.

The incremental version, has a  one-time pre-computation step in which
the program is sliced with  respect to a \emph{default criterion} that
is same for all  programs.  The result of this step  is converted to a
set  of  automata, one  for  each  expression  in the  program.   This
completes the pre-computation step.  To decide whether a expression is
in the  slice for a given  slicing criterion, we simply  intersect the
slicing criterion with the  automaton corresponding to the expression.
If the result  is the empty set, the expression can be  removed from the
slice.

The main contributions of this paper are as follows:
\begin{enumerate}
\item We propose a view of the  slicing criterion in terms of a notion
  called {\em  demand} (Section~\ref{sec:demand})  and  formulate the
  problem of  slicing as  one of  propagating the  demand on  the main
  expression to all the sub-expressions  of the program.  The analysis
  for this is precise because it  keeps the information at the calling
  context  separate.   However it  attempts  to  reduce the  attendant
  inefficiency through  the use of function  summaries. The difficulty
  of creating function summaries in a polyvariant analysis, especially
  when the  domain of analysis  is unbounded,  has been pointed out
  in \cite{reps96}.
\item  Our   formulation  (Section~\ref{sec:grammar-formulation})
  allows us to derive an incremental version of slicing algorithm
  that factors  out computations  common to all  slicing criteria
  (Section~\ref{sec:incr-analysis})     and      re-uses     these
  computations.  To  the best  of our knowledge,  the incremental
  version of slicing in this form has not been attempted before.
\item  We  have proven  the  correctness  of  the  incremental  slicing
  algorithm   with  respect   to   the  non-incremental   version
  (Section~\ref{sec:incremental-proofs}).     
\item We  have implemented a  prototype slicer for  a first-order
  version of Scheme (Section~\ref{sec:exp-result}).  We have also
  extended   the   implementation    to   higher-order   programs
  (Section~\ref{sec:higher-order}) by converting such programs to
  first-order             using             \emph{firstification}
  techniques~\cite{Mitchell:2009},    slicing   the    firstified
  programs using our slicer, and  then mapping the sliced program
  back   to  the   higher-order   version.   The   implementation
  demonstrates the expected benefits  of incremental slicing: the
  incremental step  is one to four orders  of magnitude faster  than the
  non-incremental version.
\end{enumerate}

\begin{figure}
\newcommand{\isdef}{\ensuremath{\!::=\!\!}}
  \renewcommand{\arraystretch}{.95}
  $ \begin{array}{@{}r@{\ }c@{\ }l@{}}
    p \in \mathit{Prog} & \isdef & d_1 \ldots d_n \,\, e_\mainpgm
    \hspace{5.5em} \mbox{\em --- program}\\ & \\
    d \in \mathit{Fdef} &
    \isdef & (\DEFINE\,\, (f\,\, x_1 \,\, \ldots \,\,x_n)\,\, e)
    \hspace{1.5em} \mbox{\em --- function definition} \\ & \\
    e \in
    \mathit{Expr} & \isdef &
    \left\{\begin{array}{@{}ll@{\hspace{3.6em}}l} (\SIF\,\, x\,\,
    e_1\,\, e_2) && \mbox{\em --- conditional} \\ (\LET\,\, x
    \leftarrow s\,\, \IN\,\, e) && \mbox{\em --- let binding}
    \\ (\SRETURN\,\, x) && \mbox{\em --- return from function}
    \end{array}\right. \\ &\\
    s \in \mathit{App} & \isdef &
    \left\{\begin{array}{@{}l@{}r@{\hspace{1em}}l} k & $\NIL$&
    \mbox{\em --- constants}\\ (\CONS\,\, x_1\,\, x_2) && \mbox{\em
      --- constructor} \\ (\CAR\,\, x) & (\CDR\,\, x) & \mbox{\em ---
      selectors} \\ (\NULLQ\,\, x) & (\PRIM\,\, x_1\,\, x_2) &
    \mbox{\em --- tester/generic-arithmetic} \\
    \multicolumn{2}{@{}l}{(f\,\, x_1\,\,\ldots\,\, x_n)} &
    \mbox{\em --- function application}
    \end{array}\right.
    \\
\end{array}$
\caption{The syntax of our language\label{fig:lang-syntax}}
\end{figure}

\section{The target language---syntax and semantics}
\label{sec:defs}
Figure~\ref{fig:lang-syntax}  shows the  syntax of  our language.
For   ease  of   presentation,  we   restrict  the   language  to
Administrative Normal Form (ANF)~\cite{chakravarty03perspective}.
In this form,  the arguments to functions can  only be variables.
To  avoid  dealing  with  scope-shadowing,  we  assume  that  all
variables  in  a program  are  distinct.   Neither of  these  two
restrictions affect the expressibility of our language.  In fact,
it is a simple matter to transform the pure subset of first order
Scheme  to our  language,  and  map the  sliced  program back  to
Scheme.  To refer to an expression $e$, we may annotate it with a
label $\pi$ as $\pi\!:\!e$; however the  label is not part of the
language.  To keep  the description simple, we  shall assume that
each program has its own unique set of labels.  In other words, a
label  identifies both  the program  point and  the program  that
contains it.

A  program in  our language  is a  collection of  function definitions
followed by  a main expression denoted  as $e_\mainpgm$.  Applications
(denoted  by   the  syntactic  category  $\mathit{App}$)   consist  of
functions   or   operators    applied   to   variables.    Expressions
($\mathit{Expr}$) are either an $\SIF$ expression, a $\LET$ expression
that evaluates an application and binds the result to a variable, or a
$\SRETURN$ expression.  The $\SRETURN$ keyword is used to mark the end
of a function so as to initiate appropriate semantic actions during 
execution.  The distinction between  expressions and applications will
become important while specifying the semantics of programs.


\begin{figure}
  \[\begin{array}{|@{\ }c@{\ }|c|c|}\hline
  \mbox{Premise}& \mbox{Transition} & \mbox{Rule Name} \\\hline\hline
  &\rho, \heap, \pp{\pi}{k} \rightsquigarrow \heap, k& \mbox{\sc const}
  \\ \hline
  \rho(x) \in \mathbb{N} \andalso \rho(y)\in \mathbb{N}&
  \rho, \heap, \pp{\pi}{(\PRIM\,\, x\ y)}
  \rightsquigarrow \heap, \rho(x) + \rho(y)&
  \mbox{\sc prim}
  \\ \hline
  \heap(\rho(x)) = (v_1, v_2) &
  \rho, \heap, \pp{\pi}{(\CAR\,\, x)} \rightsquigarrow \heap, v_1 &
  \mbox{\sc car}
  \\ \hline
  \heap(\rho(x)) = (v_1, v_2) &
  \rho, \heap, \pp{\pi}{(\CDR\,\, x)} \rightsquigarrow \heap, v_2 &
  \mbox{\sc cdr}     
  \\ \hline
  \mbox{$\ell \not\in \mbox{dom}(\heap)$ is a fresh location}&
  \rho, \heap, \pp{\pi}{(\CONS\,\, x\ y)} \rightsquigarrow \heap[\ell\mapsto(\rho(x),\rho(y))], \ell &
  \mbox{\sc cons}     
  \\ \hline
  \rho(x) \in \mathbb{N}\setminus\{0\}&
  \rho,\stack,
      \heap, \pp{\pi}{(\SIF\ x\ \pp{\pi_2}{e_1}\ \pp{\pi_3}{e_2})}
      \longrightarrow \rho,\stack, \heap, e_1 &
  \mbox{\sc if-true}     
  \\ \hline
  \rho(x) = 0&
  \rho,\stack, \heap, \pp{\pi}{(\SIF\ x\ \pp{\pi_2}{e_1}\ \pp{\pi_3}{e_2})}
  \longrightarrow \rho,\stack, \heap, e_2 &
  \mbox{\sc if-false}     
  \\ \hline
  \rho(x)\not= \NIL&
  \rho, \heap, \pp{\pi}{(\NULLQ\,\, x)} \rightsquigarrow \heap, 0&
  \mbox{\sc null-true}     
  \\ \hline
  \rho(x) = \NIL&
  \rho, \heap, \pp{\pi}{(\NULLQ\,\, x)} \rightsquigarrow \heap, 1&
  \mbox{\sc null-false}     
  \\ \hline
  \begin{array}{@{\ }l@{\ }}
    \mbox{$s$ is $\pp{\pi}{(f\,\, \pp{\pi_1}{y_1} \ldots \pp{\pi_n}{y_n})}$} \\
    \mbox{$f$ is $(\DEFINE\ (f\ z_1\ \ldots\ z_n)\ \pp{\pi_f}{e_f})$}
  \end{array}
  &
  \rho,\stack, \heap, \pp{\pi}{(\LET\,\, x \leftarrow
    \pp{\pi_1}{s}\,\, \IN\,\, \pp{\pi_2}{e})} \longrightarrow
      [\myvec{z} \mapsto \rho(\myvec{y})], \,(\rho,x,e)\bullet\stack,\,
      \heap, e_f&   
  \mbox{\sc let-fncall}     
  \\ \hline
  \begin{array}{@{\ }l@{\ }}
    \rho, \heap, \pp{\pi_1}{s} \rightsquigarrow \heap', v\\
    \mbox{$s$ is not $\pp{\pi}{(f\,\, \pp{\pi_1}{y_1} \ldots
        \pp{\pi_n}{y_n})}$ }
  \end{array}&
  \rho,\stack, \heap,
      \pp{\pi}{(\LET\,\, x \leftarrow \pp{\pi_1}{s}\,\, \IN\,\,
        \pp{\pi_2}{e})} \longrightarrow \rho[x \mapsto v], \stack,
      \heap', \pp{\pi_2}{e}&
  \mbox{\sc let-nonfn}     
  \\ \hline
  &\rho, \,(\rho',x',e')\bullet\stack,\,
  \heap, \pp{\pi}{(\SRETURN\ x)} \longrightarrow \rho'[x' \mapsto
    \rho(x)], \stack, \heap, e'&
  \mbox{\sc return}     
  \\ \hline
  \end{array}\]
\caption{The semantics of our language\label{fig:syntax-sematics}}
\end{figure}

\subsection{Semantics}
We now present  the operational semantics for  our language. This
is   largely  borrowed   from~\cite{asati14,  lazyliveness} and we
include it here for  completeness.  We  start with  the domains
used by the semantics:
\[
\begin{array}{rlcl@{\hspace{2em}}l}
v: & \mathit{Val} &=& \mathbb{N} + \{\NIL\} + \mathit{Loc}& \hspace{-0.5cm}\mbox{-- Values}\\ \rho: & \mathit{Env}
&=&\mathit{Var} \rightarrow {Val} &
\hspace{-0.5cm}\mbox{-- Environment} \\ \heap: & \mathit{Heap} &
=&\mathit{Loc} \rightarrow (Val\times Val + \{empty\}) &\hspace{-0.5cm}
\mbox{-- Heap}
\end{array}
\]

A value  in our language  is either a  number, or the  empty list
denoted by  $\NIL$, or a  location in the heap.  The  heap maps
each location  to a pair of values denoting a cons cell. Heap
locations can also be empty. Finally, an environment is a mapping
from variables  to values.

The    dynamic   aspects     of    the    semantics,   shown    in
Figure~\ref{fig:syntax-sematics},   are   specified  as   a   state
transition system. The semantics of applications $s$ are given by
the judgement  form $\rho, \heap, s  \rightsquigarrow \heap', v$,
and those for expressions $e$ by the form $\rho, \stack, \heap, e
\rightarrow \rho', \stack', \heap', e'$. Here $\stack$ is a stack
consisting of  continuation frames  of the  form $(\rho,  x, e)$.
The frame $(\rho,  x, e)$ signifies that if  the current function
returns a value $v$, the next  expression to be evaluated is $e$,
and the  environment for this  evaluation is $\rho$  updated with
the   variable  $x$   bound   to  $v$.    The   start  state   is
$(\{\}_\rho,[\,]_\stack,    \{\}_\heap,    e_\mainpgm)$,    where
$\{\}_\rho$ is the empty  environment, $[\,]_\stack$ is the empty
stack,  and   $\{\}_\heap$  is  the  empty   heap.   The  program
terminates successfully  with result value $\rho(x)$  on reaching
the halt state $(\rho,[\,]_\stack, \heap, (\SRETURN\ x))$. We use
the  notation  $\rho[x  \mapsto  v]$ to  denote  the  environment
obtained by  updating $\rho$ with  the value  for $x$ as  $v$. We
also use $[\myvec{x} \mapsto \myvec{v}]$ to denote an environment
in which each $x_i$ has the value $v_i$.

\section{Demand}
\label{sec:demand}
We     now     connect     slicing    with     a     notion     called
\emph{demand}. A demand  on an expression represents the
set of paths that the context  of the expression \emph{may} explore of
the  value  of   the  expression.   A  demand  is   represented  by  a
prefix-closed set of  strings over $\demandlang$.  Each  string in the
demand, called an \emph{access path},  represents a traversal over the
heap.  $\acar$  stands for  a single-step traversal  over the  heap by
dereferencing the  $\CAR$ field  of a  cons cell.   Similarly, $\acdr$
denotes the dereferencing of the $\CDR$ field of a cons cell.
 
As an  example, a demand  of $\{\epsilon, \acdr, \acdr\acar\}$  on the
expression $(\CONS~x~y)$  means its  context \emph{may} need  to visit
the  \CAR\ field  of  $y$ in  the heap  (corresponding  to the  string
\acdr\acar\ in the demand).  The  example also illustrates why demands
are prefix-closed---the \CAR\  field of $y$ cannot  be visited without
visiting  first  the  cons  cell  resulting  from  the  evaluation  of
$(\CONS~x~y)$   (represented  by   $\epsilon$)  and   then  the   cell
corresponding to $y$ (represented by \acdr).  The absence of \acar\ in
the demand  also indicates that  $x$ is {\em definitely  not} visited.
Notice that  to meet the  demand $\{\epsilon, \acdr,  \acdr\acar\}$ on
$(\CONS~x~y)$,  the  access  paths  $\{\epsilon, \acar\}$  has  to  be
visited starting  from $y$.  Thus we  can think of $(\CONS~x~y)$  as a
\emph{demand transformer} transforming  the demand $\{\epsilon, \acdr,
\acdr\acar\}$ to the demand $\{\epsilon, \acar\}$ on $y$ and the empty
demand (represented by $\emptyset$) on $x$.

The slicing  problem is now  modeled as follows.  Viewing  the slicing
criterion  (also   a  set   of  strings   over  $\demandlang$)   as  a
demand\footnote{supplied by a context that is external to the program}
on the  main expression  $e_\mainpgm$, we compute  the demand  on each
expression in the program.  If the demand on a expression turns out to
be $\emptyset$,  the expression does  not contribute to the  demand on
$e_\mainpgm$ and can be removed from  the slice.  Thus the solution of
the slicing problem lies in computing a demand transformer that, given
a  demand on  $e_\mainpgm$, computes  a \emph{demand  environment}---a
mapping of each expression (represented by its program point $\pi$) to
its  demand.  We  formulate  this computation  as  an analysis  called
\emph{demand analysis}.

We use $\sigma$ to represent  demands and $\alpha$ to represent access
path.  Given  two access paths  $\alpha_1$ and $\alpha_2$, we  use the
juxtaposition  $\alpha_1\alpha_2$ to  denote their  concatenation.  We
extend this  notation to a concatenate  a pair of demands  and even to
the  concatenation  of a  symbol  with  a demand:   $\sigma_1\sigma_2$
denotes  the   demand  $\{\alpha_1\alpha_2\mid  \alpha_1   \in  \sigma
~\mbox{and}~ \alpha_2 \in \sigma_2\}$ and $\acar\sigma$ is a shorthand
for $\{\acar\alpha \mid \alpha \in \sigma\}$.

\begin{figure*}[t]
\begin{align*}
  \mathit{\A}(\pi\!\!:\!\!\kappa, \demand, \DT) & = \{\pi
  \mapsto \demand\} , \ \text{for constants
    including} \ \mNIL \\
  \mathit{\A}(\pi\!\!:\!\!(\mNULLQ \ \pi_1\!\!:\!\!x), \demand, \DT) &= \{\pi_1 \mapsto \mbox{if}~ \sigma \ne \emptyset~\mbox{then}~\{\epsilon\}~\mbox{else}~ \emptyset,\; \pi \mapsto \demand\}\\
  \mathit{\A}(\pi\!\!\!:\!\!(\PRIM \ \pi_1\!\!:\!\!x \ \pi_2\!\!:\!\!y), \demand, \DT) & = \{\pi_1 \mapsto \mbox{if}~ \sigma \ne \emptyset~\mbox{then}~\{\epsilon\}~\mbox{else}~    \emptyset, \;\;\pi_2 \mapsto \mbox{if}~ \sigma \ne \emptyset~\mbox{then}~\{\epsilon\}~\mbox{else}~ \emptyset,  \;\;\pi \mapsto \demand\}\\
  \mathit{\A}(\pi\!\!:\!\!:(\mCAR\;\pi_1\!\!:x), \demand,    \DT) & = \{\pi_1 \mapsto \mbox{if}~ \sigma \ne    \emptyset~\mbox{then}~\{\epsilon\}\cup \acar\sigma    ~\mbox{else}~ \emptyset, \;\;\pi \mapsto    \demand\}\\
  \mathit{\A}(\pi\!\!:\!\!:(\mCDR\;\pi_1\!\!:x), \demand,    \DT) & = \{\pi_1 \mapsto \mbox{if}~ \sigma \ne    \emptyset~\mbox{then}~\{\epsilon\}\cup \acdr\sigma    ~\mbox{else}~ \emptyset, \;\;\pi \mapsto  \demand\}\\
  \mathit{\A}(\pi\!\!:\!\!(\mCONS \ \pi_1\!\!:\!\!x    \ \pi_2\!\!:\!\!y), \demand, \DT) & = \{\pi_1\!  \mapsto    \{\alpha \mid \acar\alpha \in \sigma\}, \pi_2\! \mapsto    \{\alpha \mid \acdr\alpha \in \sigma\}, \;\;\pi    \mapsto \demand\}\\
  \mathit{\A}(\pi\!\!:\!\!(f \ \pi_1\!\!:\!\!y_1    \ \cdots \ \pi_n\!\!:\!\!y_n), \demand, \DT) & =    \bigcup_{i=1}^n \{\pi_i \mapsto \DT^i_f(\demand)\} \cup \{\pi    \mapsto \demand \}\\
  \mathcal{D}(\pi\!\!:\!\!(\mRETURN    \ \pi_1\!\!:\!\!x), \demand, \DT) & = \{\pi_1 \mapsto    \demand, \pi \mapsto \demand\}    \\
  \mathcal{D}(\pi\!\!:\!\!(\mIF \ \pi_1\!\!:\!\!x~e_1~e_2),    \demand, \DT) & = \mathcal{D}(e_1, \demand, \DT) \cup    \mathcal{D}(e_2, \demand, \DT) ~\cup  \;\{ \pi_1    \mapsto \mbox{if}~ \sigma \ne    \emptyset~\mbox{then}~\{\epsilon\}~\mbox{else}~ \emptyset, \; \pi \mapsto \demand \}
    \\
\mathcal{D}(\pi\!\!:\!\!(\mLET \ \pi_1\!\!:\!\!x   \leftarrow s \ \mIN \ e),\demand,\DT) & = {\A}(s, \mathop{\cup}_{\pi \in \Pi}\mbox{DE}(\pi), \DT) \; \cup \{\pi \mapsto
    \demand\} \\
&~~~~\text{where} \ \mbox{DE} = \mathcal{D}(e, \demand, \DT), ~\text{and}~~\Pi~\text{represents all occurrences of $x$ in $e$},
\end{align*}\\[8pt]
\begin{minipage}[l]{\textwidth}
\infrule[demand-summary] { \forall f, \forall i, \forall \sigma: \;\;
  \mathcal{D}(e_\mathit{f},\sigma,\DT) = \mbox{DE}, \DT^{i}_{\mathit f}
  = \mathop{\bigcup}_{\pi \in \Pi}~\mbox{DE}(\pi) } { \mathit{df_1} \ldots \mathit{df_k} \len
  \DT}  \medskip
\end{minipage}\begin{minipage}{.2\textwidth}{}
\end{minipage}\\
\begin{minipage}{.8\textwidth}
  where
  $(\DEFINE\ (f\ z_1\ \ldots\ z_n)\ \ e_\mathit{f})$
  is one of $\mathit{df_1}$ \ldots $\mathit{df_k}$,
  $1\leq i\leq n$,  and $\Pi$ represents all occurrences of
  $z_i$ in $e_f$
\end{minipage}
\caption{Demand Analysis}\label{fig:demand-analysis}
\end{figure*}

\subsection{Demand Analysis}
Figure~\ref{fig:demand-analysis}   shows  the   analysis.   Given   an
application  $s$  and   a  demand  $\sigma$,  \A\   returns  a  demand
environment that maps expressions of  $s$ to their demands.  The third
parameter to \A, denoted \DT, represents context-independent summaries
of the functions in the program, and will be explained shortly.

Consider the rule for the selector  $\CAR$.  If the demand $\sigma$ on
$(\CAR~x)$ is $\emptyset$, then no part  of the value of $(\CAR~x)$ is
visited  and the  demand  on  $x$ is  also  $\emptyset$.  However,  if
$\sigma$  is  non-empty,  the  context  of  $(\CAR~x)$  has  to  first
dereference the value of $x$ using the $\CAR$\ field and then traverse
the paths represented by $\sigma$. In  this case, the demand on $x$ is
the  set consisting  of  $\epsilon$ (start  at the  root  of $x$)  and
$\acar\sigma$ (dereference  using \CAR\  and then  visit the  paths in
$\sigma$).  On   the  other  hand,   the  rule  for   the  constructor
\CONS\  works   as  follows:   To  traverse  the   path  $\acar\alpha$
(alternately $\acdr\alpha$)  starting from the root  of $(\CONS~x~y)$,
one has to traverse the path $\alpha$ starting from $x$ (or $y$).

Since  $(\mNULLQ~x)$  only visits  the  root  of  $x$ to  examine  the
constructor,  a non-null  demand  on $(\mNULLQ~x)$  translates to  the
demand $\epsilon$ on $x$.  A  similar reasoning also explains the rule
for $(\PRIM~x~y)$.  Since, both $x$ and  $y$ evaluate to integers in a
well typed program,  a non-null demand on  $(\PRIM~x~y)$ translates to
the demand $\epsilon$ on both $x$ and $y$.

The  rule for  a  function  call uses  a  third  parameter $\DT$  that
represents the summaries of all functions in the program.  $\DT$\ is a
set  of   context-independent  summaries,  one  for   each  (function,
parameter)  pair  in  the   program.   $\DT_{\mathit  f}^{\mathit  i}$
represents a transformation that describes  how any demand $\sigma$ on
a call to  $f$ is transformed into the demand  on its $i$th parameter.
$\DT$\ is specified by the  inference rule {\sc demand-summary}.  This
rule gives  a fixed-point property  to be satisfied by  $\DT$, namely,
the demand  transformation assumed  for each  function in  the program
should be  the same as  the demand transformation calculated  from the
body of the function.  Given $\DT$,  the rule for the function call is
obvious.  Notice that the demand  environment for each application $s$
also includes the demand on $s$ itself apart from its sub-expressions.
Operationally,  the  rule {\sc  demand-summary}  is  converted into  a
grammar (Section~\ref{sec:grammar-formulation})  that is parameterized
with  respect  to  a  placeholder  terminal  representing  a  symbolic
demand. The language  generated by this grammar is  the least solution
satisfying  the rule.   The  least solution  corresponds  to the  most
precise slice.

We finally discuss  the rules for expressions  given by $\mathcal{D}$.
The  rules for  $\mRETURN$\ and  $\mIF$\  are obvious.   The rule  for
$(\mLET \ x \leftarrow s \ \mIN \ e)$ first uses $\sigma$ to calculate
the demand environment DE of the  $\mLET$-body $e$.  The demand on $s$
is the union of  the demands on all occurrences of $x$  in $e$.  It is
easy  to see  by  examining the  rules that  the  analysis results  in
demands    that     are    prefix-closed.    More     formally,    let
$\mbox{DE}_{\sigma}$  be the  demand  environment  resulting from  the
analysis of a program for a  demand $\sigma$.  Then, for an expression
$\pi\!\!:\!e$   in   the  program,
$\mbox{DE}_{\sigma}(\pi)$ is prefix closed.

\section{Computing Context-Independent Function Summaries}
\label{sec:grammar-formulation}
A slicing  method used for, say,  debugging needs to be  as precise as
possible to avoid  false errors.  We therefore choose  to analyze each
function call separately  with respect to its calling  context. We now
show how  to obtain  a context-independent  summary for  each function
definition  from  the rule  {\sc  demand-summary}.   Recall that  this
summary is  a function that transforms  any demand on the  result of a
call to demands  on the arguments.  A convenient way  of doing this is
to express how a \emph{symbolic} demand  is transformed by the body of
a function.   Summarizing the function  in this way has  two benefits.
It helps us  to propagate a demand across several  calls to a function
without analyzing  its body each  time.  Even more importantly,  it is
the key to our incremental slicing method.

However, notice  that the rules of  demand analysis requires us  to do
operations that cannot be done on  a symbolic demand. The \CONS\ rule,
for example is defined in terms  of the set $\{\alpha \mid \acar\alpha
\in  \sigma\}$.  Clearly  this  requires  us to  know  the strings  in
$\sigma$.  Similarly, the \mIF\ rule requires to know whether $\sigma$
is  $\emptyset$.   The way  out  is  to  treat these  operations  also
symbolically.  For this we introduce  three new symbols \bcar, \bcdr\/
and  \clazy\/,  to  capture   the  intended  operations.   If  \acar\/
represents  selection using  \CAR, \bcar\/  is intended to
represent  a use as the left argument of  \CONS. Thus  \bcar\acar\/
should  reduce  to the  empty  string  $\epsilon$. Similarly  \clazy\/
represents  the  symbolic transformation  of  any  non-null demand  to
$\epsilon$  and  null  demand  to itself.   These  transformation  are
defined  and  also  made   deterministic  through  the  simplification
function $\mathcal{S}$.
\begin{align*}
\simpl{\epsilonset} &= \epsilonset \\
\simpl{\acar\sigma} &= \acar\simpl{\sigma}\\
\simpl{\acdr\sigma} &= \acdr\simpl{\sigma}\\
\simpl{\bcar\sigma} &= \{\alpha \mid \acar\alpha \in
\simpl{\sigma}\}\\
\simpl{\bcdr\sigma} &= \{\alpha \mid \acdr\alpha \in
\simpl{\sigma}\}\\
\simpl{\clazy\sigma} &=
  \left\{
  \begin{array}{ll}
    \emptyset&\mbox{if}~\simpl{\sigma} = \emptyset\\ \{\epsilon\} &
    \mbox{otherwise}
  \end{array}\right.\\
  \simpl{\sigma_1 \cup \sigma_2} &= \simpl{\sigma_1} \cup \simpl{\sigma_2}
\end{align*}
Notice that  \bcar\  strips the leading  \acar\ from the string following
it, as  required by the rule for \CONS. Similarly, 
\clazy\ examines  the string following  it and replaces it  by 
$\emptyset$  or  $\{\epsilon\}$; this is  required  by  several rules.   The
$\mathcal{A}$  rules for  \CONS\ and  \CAR\ in  terms of  the new
symbols are:
\begin{align*} 
&\A(\pi\!\!:\!(\CONS~\pi_1\!\!:\!x~\pi_2\!\!:\!y),\sigma,\DT)
  = \{\pi_1 \mapsto \bcar\sigma, \pi_2 \mapsto \bcdr\sigma, \pi
  \mapsto \sigma\} ~
  \\ &\A(\pi\!\!:\!(\CAR~\pi_1\!\!:\!x),\sigma,\DT) = \{
  \pi_1 \mapsto \clazy\sigma \cup \acar\sigma, \pi \mapsto \sigma
  \}
\end{align*}
and the $\mathcal{D}$ rule for \SIF\ is:
\begin{align*}
\mathcal{D}(\pi\!\!:\!(\SIF~\pi_1\!\!:\!x~e_1~e_2),\sigma,\DT) &= \mathcal{D}(e_1,
\demand, \DT) \cup \mathcal{D}(e_2, \demand, \DT) ~\cup\\ &~~~~\{
\pi_1 \mapsto \mbox{if}~ \sigma \ne
\emptyset~\mbox{then}~\{\epsilon\}~\mbox{else}~ \emptyset,
\\ &~~~~\pi \mapsto \sigma \}
\end{align*}
The   rules  for   \CDR,  \PRIM\   and  \NULLQ\   are  also   modified
similarly. Now the demand summaries  can be obtained symbolically with
the new  symbols as markers  indicating the operations that  should be
performed string following it.  When the final demand environments are
obtained with the given slicing criterion acting a concrete demand for
the  main  expression  $e_\mainpgm$,  the symbols  \bcar,  \bcdr\  and
\clazy\   are    eliminated   using   the    simplification   function
$\mathcal{S}$.

\subsection{Finding closed-forms for the summaries $\DT$}
Recall  that $\DT_{f}^i$  is a  function that  describes how  the
demand on  a call  to $f$  translates to  its $i$th  argument.  A
straightforward translation  of the {\sc demand-summary}  rule to
obtain $\DT_{f}^i$ is as follows:  For a symbolic demand $\sigma$
compute  the  the  demand  environment  in  $e_f$,  the  body  of
$f$. From this calculate the demand on the $i$th argument of $f$,
say $x$. This  is the union of demands of  all occurrences of $x$
in the body  of $f$. The demand on the  $i$th argument is equated
to $\DT_{f}^i(\sigma)$.  Since the body may  contain other calls,
the demand analysis within $e_f$ makes use of $\DT$ in turn. Thus
our equations  may be recursive. On the whole, $\DT$ corresponds
to a set of equations, one for each argument of each function.
The reader can verify that $\Lf{\linecharcount}{2}{\sigma}$ in
our running example is:
 \begin{eqnarray*}
   \Lf{\linecharcount}{2}{\sigma} &=& \bcar\sigma \cup \clazy
   \Lf{\linecharcount}{2}{\sigma}
 \end{eqnarray*} 

As  noted  in~\cite{reps96},  the   main  difficulty  in  obtaining  a
convenient function  summary is to  find a closed-form  description of
$\Lf{\linecharcount}{2}{\sigma}$    instead     of    the    recursive
specification.   Our solution  to the  problem lies  in the  following
observation: Since we know that the demand rules always prefix symbols
to the argument  demand $\sigma$, we can  write $\Lf{f}{i}{\sigma}$ as
\concatenate{\Lfun{f}{i}}{\sigma},  where $\Lfun{f}{i}$  is  a set  of
strings over the  alphabet $\{\acar, \acdr, \bcar,  \bcdr, \clazy \}$.
The modified equations after doing this substitution will be,
 \begin{eqnarray*}
   \concatenate{\Lfun{\linecharcount}{2}}{\sigma} &=& \bcar\sigma \cup \clazy \concatenate{\Lfun{\linecharcount}{2}}{\sigma}
 \end{eqnarray*}
 Thus, we have,
 \begin{eqnarray*}
   \Lf{\linecharcount}{2}{\sigma} &=&
   \concatenate{\Lfun{\linecharcount}{2}}{\sigma}\\
   \mbox{where }\Lfun{\linecharcount}{2} &=& \{\bcar\}
   \cup \clazy\Lfun{\linecharcount}{2}
 \end{eqnarray*}

\subsection{Computing the demand environment for the function
  bodies} The demand environment for a function body $e_\mathit{f}$ is
calculated  with respect  to a  concrete  demand.  To  start with,  we
consider  the main  expression $e_\mainpgm$  as  being the  body of  a
function $\mainpgm$, The  demand on $e_\mainpgm$ is  the given slicing
criterion. Further,  the concrete  demand on  a function  $f$, denoted
$\sigma_f$, is the union of the demands at all call-sites of $f$.  The
demand  environment  of a  function  body  $e_f$ is  calculated  using
$\sigma_f$. If there is a call to $g$ inside $e_f$, the demand summary
$\DT_g$ is used  to propagate the demand across  the call.  Continuing
with  our example,  the union  of the  demands on  the three  calls to
\linecharcount\ is the slicing criterion.  Therefore the demand on the
expression at program point $\pi_1$ is given by
\begin{eqnarray}
  \begin{array}{rcl}
  \Demand{\pi_1} &=& \Lfun{\linecharcount}{2}{\sigma_{\linecharcount}}\\
  \Lfun{\linecharcount}{2}&=& \{\bcar\} \cup \clazy
   \Lfun{\linecharcount}{2}\\
   \sigma_{\linecharcount} &=& \mbox{slicing criterion}
  \end{array}\label{eqn:rec-dpi1}
 \end{eqnarray}
At the end of this step, we shall have (i) A set of equations
defining the demand summaries $\DT_{\mathit {f}}^{i}$ for each argument of each
function, (ii) Equations specifying the demand $\Demand{\pi}$ at
each program point $\pi$, and (iii) an equation for each
concrete demand $\sigma_f$ on the body of each function $f$.

\subsection{Converting analysis equations to grammars}

Notice that the equations  for \Lfun{\linecharcount}{2} are still
recursive.  However,  Equation~\ref{eqn:rec-dpi1} can  also be
viewed as a  grammar with $\{\acar,\acdr,\bcdr,\bcar,\clazy\}$ as
terminal     symbols      and     \Lfun{\linecharcount}{2},
\Demand{\pi_1}  and  $\sigma_{\linecharcount}$ as  non-terminals.
Thus finding  the solution to  the set of equations  generated by 
the demand analysis reduces to  finding the language generated by
the  corresponding grammar.   The original  equations can  now be
re-written   as    grammar   rules as shown below: 
\begin{align} 
  \label{eqn:grammar-fragment}
  \begin{split}
   \var{\Demand{\pi_1}} &\rightarrow \var{\Lfun{\linecharcount}{2}}\var{\sigma_{\linecharcount}}\\
      \var{\Lfun{\linecharcount}{2}} &\rightarrow \bcar \mid
      \clazy\; \var{\Lfun{\linecharcount}{2}}      
\\
   \var{\sigma_{\linecharcount}} &\rightarrow \mbox{slicing
     criteria}
  \end{split}
\end{align}
 Thus  the question  whether  the expression  at  $\pi_1$ can  be
 sliced  for the  slicing criterion  $\sigma_{\linecharcount}$ is
 equivalent to asking whether the language  $\simpl{\lang{\Demand{\pi_1}}}$ is
 empty.  In fact, the simplification process $\mathcal{S}$ itself
 can  be captured  by  adding  the following  set of  five
 unrestricted  productions named $\mathit{unrestricted}$  and
 adding   the   production   $\var{\Demandp{\pi_1}}   \rightarrow
 \var{\Demand{\pi_1}}\$$ to the grammar generated earlier.
 \begin{eqnarray*} 
   \bcar\acar \rightarrow \epsilon &&
   \bcdr\acdr \rightarrow \epsilon\\
   \clazy\$   \rightarrow \epsilon&&
   \clazy\acar \rightarrow \clazy\\
   \clazy\acdr \rightarrow \clazy
 \end{eqnarray*}
 The set  of five unrestricted productions  shown are independent
 of the  program being sliced  and the slicing criterion.
 The symbol \$ marks the end of a sentence  and is required to capture
 the $\clazy$ rule  correctly.   

 We now generalize: Assume that  $\pi$ is the program point associated
 with  an expression  $e$.  Given  a slicing  criterion $\sigma$,  let
 $G_{\pi}^\sigma$  denote  the  grammar $(N,~T,~  P_{\pi}^\sigma  \cup
 \mathit{unrestricted}\cup         \{\Demandp{\pi}         \rightarrow
 \Demand{\pi}\$\},~  \var{\Demandp{\pi}})$.  Here  $T$ is  the set  of
 terminals   $\{\acar,   \acdr,   \bcar,  \bcdr,   \clazy,   \$   \}$,
 $P_{\pi}^\sigma$  is the  set  of  context-free productions  defining
 $\Demand{\pi}$,  the  demand  on   $e$  (as  illustrated  by  example
 \ref{eqn:grammar-fragment}).   $N$  contains the non-terminals  of
 $P_{\pi}^\sigma$ and additionally includes the special non-terminal
 $\Demandp{\pi}$.
   As mentioned earlier, given  a slicing criterion $\sigma$, the
   question of  whether the expression  $e$ can be sliced  out of
   the  containing program  is equivalent  to asking  whether the
   language $\lang{G_{\pi}^\sigma}$ is empty.   We shall now show
   that this problem is undecidable.   

   \begin{theorem}
     Given a program point  $\pi$ and slicing criterion $\sigma$,
     the  problem  whether  $\lang{G_{\pi}^\sigma}$ is  empty  is
     undecidable.
\end{theorem}

   \begin{proofoutline}  
     Recollect that the  set of demands on an  expression, as obtained
     by  our analysis,  is  prefix closed.  Since  the grammar  always
     includes       production      $\{\Demandp{\pi}       \rightarrow
     \Demand{\pi}\$\}$,  $\lang{G_{\pi}^\sigma}$ is  non-empty if  and
     only if it contains $\$$ (i.e.  empty string followed by the $\$$
     symbol).  We therefore  have to show that  the equivalent problem
     of  whether   $\$$   belongs  to  $\lang{G_{\pi}^\sigma}$  is
     undecidable.

     Given a Turing machine and a string $\alpha \in (0+1)^\ast$,
     the  proof   involves  construction   of  a  grammar   $G  =
     (N\cup\lbrace S, S'\rbrace,  T, P \cup \mathit{unrestricted}
     \cup \{S' \rightarrow S\$\}, S')$ with the property that the
     Turing machine halts on $\alpha$  if and only if $G$ accepts
     $\$$.   Notice  that  $P$  is a  set  of  context-free
     productions  over   the  terminal   set  $T$  and   may  not
     necessarily  be  obtainable  from   demand  analysis  of  a
     program.  However,  $G$ can be  used to construct  a program
     whose demand  analysis results  in a  grammar $G'$  that can
     used  instead  of $G$  to  replay  the earlier  proof.   The
     details   can  be   found   in  Lemmas   B.2   and  B.3   of
     \cite{lazyliveness}.
   \end{proofoutline}

\begin{figure}[t]
  \centering
  \includegraphics[width=.6\textwidth]{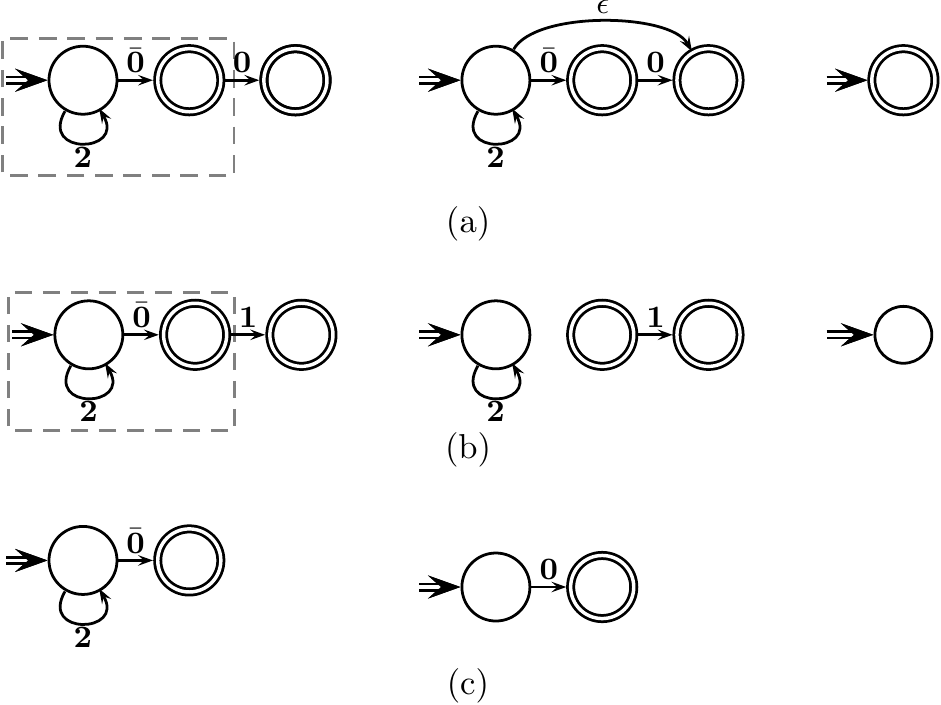}
  \caption{(a) \& (b) show the simplification of the automaton $\automaton{M}{\pi_1}{\sigma}$ for the slicing criteria $\sigma = \{\epsilon, \acar\}$ and $\sigma = \{\epsilon, \acdr\}$ respectively. (c) shows the canonical automaton $\grammar{A}{\pi_{1}}{}$ and the corresponding completing automaton $\compautomaton{A}{\pi_1}$}
  \label{fig:nfa}
\end{figure}
 
We get around  the problem of undecidability, we use  the technique of
Mohri-Nederhoff~\cite{mohri00regular}        to       over-approximate
$\grammar{P}{\pi}{\sigma}$  by a  strongly regular  grammar.  The  NFA
corresponding     to     this      automaton     is     denoted     as
$\automaton{M}{\pi}{\sigma}$.  The simplification rules can be applied
on  $\automaton{M}{\pi}{\sigma}$ without  any loss  of precision.  The
details of the simplification process are in~\cite{karkare07liveness}.

 For our running example, the grammar after demand analysis is
 already regular, and thus remains unchanged by Mohri-Nederhoff
 transformation.  The automata in Figures~\ref{fig:nfa}(a)
 and~\ref{fig:nfa}(b) correspond to the two slicing criteria
 $\sigma_{\linecharcount} = \{\epsilon, \acar\}$ and
 $\sigma_{\linecharcount} = \{\epsilon, \acdr\}$ and illustrate the 
 simplification of corresponding Mohri-Nederhoff automata $\automaton{M}{\pi_1}{\sigma_{\linecharcount}}$. It can be seen that, when the slicing
 criterion is $\{\epsilon, \acdr\}$, the language of
 $\Demand{\pi_1}$ is empty and hence $e$ can be sliced away.
 A drawback of the method outlined above  is that with a change in the
 slicing  criterion,   the  entire  process  of   grammar  generation,
 Mohri-Nederhoff   approximation   and   simplification  has   to   be
 repeated. This is likely to be inefficient for large programs.

\section{Incremental Slicing}  \label{sec:incr-analysis}
We now present an incremental algorithm which avoids the repetition of
computation when the  same program is sliced  with different criteria.
This can  be done by pre-computing  the part of the  slice computation
that is independent  of the slicing criterion.   The pre-computed part
can then be used efficiently to  slice the program for a given slicing
criterion.

In general, the pre-computation consists of three steps: (i) computing
the  demand  at  each  expression $\pi\!: e$  for  the  \emph{fixed}  slicing
criterion $\{\epsilon\}$ and applying the Mohri-Nederhoff procedure to
yield the automaton $\grammar{M}{\pi}{\{\epsilon\}}$, (ii) a step called
\emph{canonicalization}  which  applies  the simplification  rules  on
$\grammar{M}{\pi}{\{\epsilon\}}$ until the \bcar\  and \bcdr\ symbols in
the strings accepted  by the resulting automaton are only  at the end,
and,   from  this   (iii)   constructing  an   automaton  called   the
\emph{completing automaton}.  For the running example, the canonicalized and the completing
automata are shown
Figures~\ref{fig:nfa}(c).   We explain these
steps now.

  As          stated          earlier,          the          automaton
  $\grammar{M}{\pi_{1}}{\{\epsilon\}}$,  after  some  simplifications,
  gives the first automaton (the canonicalized automaton) shown in  Figure~\ref{fig:nfa}(c), which  we shall
  denote    $\grammar{A}{\pi_{1}}{}$.     It    is   clear    that    if
  $\grammar{A}{\pi_{1}}{}$  is concatenated  with a  slicing criterion
  that starts with the symbol $\acar$, the result, after simplification,
  will be non-empty.   We call a string that starts  with $\acar$ as a
  \emph{completing string} for $\grammar{A}{\pi_{1}}{}$. In this case,
  detecting a completing string was easy  because  all  strings  accepted  by
  $\grammar{A}{\pi_{1}}{}$ end  with $\bcar$.   Now consider  the second automaton in Figure~\ref{fig:nfa}(c),
  called the \emph{completing automaton}, that recognizes the language
  $\acar(\acar+\acdr)^\ast$.  This automaton recognizes all completing
  strings  for  $\grammar{A}{\pi_{1}}{}$  and  nothing else.   Thus  for  an
  arbitrary  slicing  criterion  $\sigma$, it  suffices  to  intersect
  $\sigma$  with  the  completing  automaton  to  decide  whether  the
  expression at $\pi_1$ will  be in the  slice. In fact,  it is
  enough for the  completing automaton to recognize  just the language
  $\{\acar\}$ instead of  $\acar(\acar+\acdr)^\ast$.  The reason  is that
  any slicing criterion, say $\sigma$, is prefix closed, and therefore
  $\sigma  \cap  \{\acar\}$ is  empty  if  and  only if  $\sigma  \cap
  \acar(\acar+\acdr)^\ast$  is   empty.   Our   incremental  algorithm
  generalizes this reasoning.

\subsection{Completing Automaton and Slicing}
For constructing the completing  automaton for an expression $e$,
we  saw that  it would  be convenient  to simplify  the automaton
$\grammar{M}{e}{\{\epsilon\}}$ to an extent  that all  accepted strings, after
simplification,  have  $\bcar$   and $\bcdr$  symbols
only at the end. We now give a set of rules, denoted by $\canon$,
that captures this simplification. 
\begin{align*}
\canonical{\epsilonset} &= \epsilonset\\
\canonical{\acar\sigma} &= \acar\canonical{\sigma} \\
\canonical{\acdr\sigma} &= \acdr\canonical{\sigma} \\
\canonical{\bcar\sigma} &= 
    \{\bcar \mid \epsilonset = \canonical{\sigma}\} \cup \{\alpha \mid \acar\alpha \in \canonical{\sigma}\} \\
  & \cup \{\bcar\bcdr\alpha \mid \bcdr\alpha \in \canonical{\sigma}\} \cup \{\bcar\bcar\alpha \mid \bcar\alpha \in \canonical{\sigma}\} \\
  \canonical{\bcdr\sigma} &=  
    \{\bcdr \mid \epsilonset = \canonical{\sigma}\} \cup   \{\alpha \mid \acdr\alpha \in \canonical{\sigma}\} \\
  & \cup \{\bcdr\bcdr\alpha \mid \bcdr\alpha \in \canonical{\sigma}\} \cup \{\bcdr\bcar\alpha \mid \bcar\alpha \in \canonical{\sigma}\} \\
    \canonical{\clazy\sigma} &= \clazy\canonical{\sigma} \\
    \canonical{\sigma_1 \cup \sigma_2} &= \canonical{\sigma_1} \cup \canonical{\sigma_2}
\end{align*}
$\canon$ differs from $\Simpl$  in that it accumulates continuous
run of  \bcar\ and \bcdr\  at the end  of a string.   Notice that
$\canon$,  like $\Simpl$,  simplifies its  input string  from the
right.    Here  is   an  example   of  $\canon$   simplification:
  \[\acdr\clazy\bcar\acar\acar\clazy\acar\bcdr\bcdr\acdr\bcar \csimpl
  \acdr\clazy\bcar\acar\acar\clazy\acar\bcdr\bcar \csimpl
  \acdr\clazy\acar\clazy\acar\bcdr\bcar \]
In contrast the simplification of the same string using
$\Simpl$ gives:
\[\acdr\clazy\bcar\acar\acar\clazy\acar\bcdr\bcdr\acdr\bcar
\ssimpl \acdr\clazy\bcar\acar\acar\clazy\acar\bcdr\bcdr\acdr\emptyset
\ssimpl
\acdr\clazy\bcar\acar\acar\clazy\acar\bcdr\bcar\emptyset
\ssimpl \ldots \ssimpl \emptyset
\] 

$\canon$ satisfies two important properties:
\begin{property} \label{prop:canon-1}
The result of $\canon$ always has the form $(\acar + \acdr +
\clazy)^\ast(\bcar + \bcdr)^\ast$.  Further,   if  $\sigma   \subseteq
  (\acar+\acdr+\clazy)^\ast$, then $\canonical{\sigma} = \sigma$.
\end{property}
\begin{property}\label{prop:canon-2}
  $\Simpl$      subsumes      $\canon$, i.e., 
  $\simpl{\canonical{\sigma_1}\canonical{\sigma_2}}             =
  \simpl{\sigma_1\sigma_2}$. 
\end{property}
Note that while we have defined canonicalization over a language,
the    actual    canonicalization    takes    place    over    an
automaton---specifically    the   automaton    $\grammar{M}{\pi}{}$
obtained after the Mohri-Nederhoff transformation.
The       function       {\bf      createCompletingAutomaton}       in
Algorithm~\ref{algo:comp-automata} takes  $\grammar{A}{\pi}{}$, the
  canonicalized  Mohri-Nederhoff automaton  for the  slicing criterion
  $\epsilonset$,  as  input, and  constructs the  completing automaton,
  denoted as $\compautomaton{A}{\pi}$.

   \SetAlCapSkip{0.5em}
\begin{algorithm}[t!]
  \caption{Functions  to create the completing automaton and the slicing function.\label{algo:comp-automata}}
  \DontPrintSemicolon\AlgoDontDisplayBlockMarkers\SetAlgoNoLine
  \SetAlgoNoEnd
  \SetKwData{frSet}{\ensuremath{F_{\rm fr}}}
  \SetKwData{Acomp}{\ensuremath{\compautomaton{A}{}}}
  \SetKwFunction{CreateCompletingAutomaton}{\bf createCompletingAutomaton}
  \SetKwProg{Fn}{Function}{}{}
  \SetKwData{hasBarFreePath}{hasBarFreeTransition}
  \AlgoDontDisplayBlockMarkers
  \SetKwFunction{Inslice}{\bf inSlice}
  \SetKwData{exp}{e}\SetKwData{sc}{$\sigma$}
  \Fn{\CreateCompletingAutomaton{$A$}}{
    \KwData{The Canonicalized Automaton $A = \left\langle Q, \{\acar, \acdr, \bcar, \bcdr, \clazy\}, \delta, q_0, F\right\rangle$}
    \KwResult{$\Acomp$, the completing automaton for $A$}

    $F'\leftarrow \{q_{\rm fr} \mid q_{\rm fr} \in Q,\, \hasBarFreePath(q_0, q_{\rm fr}, \delta)\}$\; 

  \tcc{\dotfill Reverse the ``bar'' transitions: directions as
    well as labels \dotfill} 
    \ForEach{(transition $\delta(q, \bcar) \rightarrow q'$)}{
      add transition $\delta'(q', \acar) \rightarrow q$
    }
    \ForEach{(transition $\delta(q, \bcdr) \rightarrow q'$)}{
      add transition $\delta'(q', \acdr) \rightarrow q$
    }
    $q'_0\leftarrow$ new state \tcc{\dotfill start state of \Acomp \dotfill}
    \ForEach{(state $q \in F$) }{
      add transition $\delta'(q'_0, \epsilon) \rightarrow q$
    }
    
    \KwRet{$\left\langle Q\cup\{q'_0\}, \{\acar, \acdr\}, \delta', q'_0, F'\right\rangle$}
  }
  \Fn{\Inslice{\exp, \sc}}{
   \KwData{expression \exp, slicing criteria \sc}
    \KwResult{Decides whether \exp should be retained in slice}
    \KwRet{$(\lang{\compautomaton{A}{\exp}} \cap\ \sc \neq \emptyset)$}
  }
\end{algorithm}


Recollect that the strings recognized by $\grammar{A}{\pi}{}$ are of the
form  $(\acar  +  \acdr  +  \clazy)^\ast(\bcar  +  \bcdr)^\ast$.   The
algorithm first  computes the set  of states reachable from  the start
state using only  edges with labels $\{\acar,  \acdr, \clazy\}$.  This
set  is  called the  {\em  frontier  set}.   It then  complements  the
automaton  and  drops  all   edges  with  $\{\acar,  \acdr,  \clazy\}$
labels. Finally,  all states in the  {\em frontier set} are  marked as
final states.
Since \automaton{A}{\pi}{} is independent
of the slicing  criteria, the {\em completing  automaton} is also
independent of the slicing criteria and needs to be computed only
once. It can be stored and  re-used whenever the program needs to
be  sliced.   To  decide  whether  $\pi\!\!:e$ can  be  sliced  out,  the
function {\bf inSlice}       described      in
Algorithm~\ref{algo:comp-automata}   just   checks   if   the
intersection      of      the     slicing      criteria      with
\lang{\compautomaton{A}{\pi}} is null.

 \subsection{Correctness of Incremental Slicing}
  \label{sec:incremental-proofs}

We now show that the incremental algorithm to compute incremental
slices  is correct.  Recall that we use the following notations: (i)  \grammar{G}{\pi}{\sigma}  is the  grammar
generated         by          demand         analysis
(Figure~\ref{fig:demand-analysis})  for an  expression  $\pi\!\!: e$ in
the   program of interest, when the slicing   criteria is
$\sigma$, (ii) \automaton{A}{\pi}{}  is   the  automaton   corresponding  to
\grammar{G}{\pi}{\epsilonset}  after Mohri-Nederhoff  transformation and
canonicalization, and (iii) \compautomaton{A}{\pi}
is the completing automaton   for $e$.
We  first show  that the  result of  the demand  analysis for  an
arbitrary slicing  criterion $\sigma$ can be decomposed as  the
concatenation  of the  demand analysis  obtained for  the fixed  slicing
criterion      $\{\epsilon\}$      and      $\sigma$      itself.


\begin{lemma}
  \label{thm:grammar-eq}
For   all  expressions   $e$  and   slicing  criteria   $\sigma$,
\lang{\grammar{G}{\pi}{\sigma}}                                   =
\concatenate{\lang{\grammar{G}{\pi}{\epsilonset}}}{\sigma}.
\end{lemma}
  
\begin{proof}
 The proof  is by  induction on the  structure of  $e$. Observe
   that    all    the    rules    of    the    demand    analysis
   (Figure~\ref{fig:demand-analysis})   add   symbols   only   as
   prefixes to the incoming  demand.  Hence, the slicing criteria
   will always appear  as a suffix of any string that is   produced by the
   grammar.  Thus, any  grammar \lang{\grammar{G}{\pi}{\sigma}} can
   be  decomposed  as  $\concatenate{\sigma'}{\sigma}$  for  some
   language $\sigma'$.  Substituting $\epsilonset$  for $\sigma$,
   we   get   $\grammar{G}{\pi}{\epsilonset}  =   \sigma'$.    Thus
   \lang{\grammar{G}{\pi}{\sigma}}                                =
   \concatenate{\lang{\grammar{G}{\pi}{\epsilonset}}}{\sigma}.
\end{proof}
Given a string $s$ over $(\bcar+\bcdr)^\ast$, we use the notation $\overline{s}$ to stand for the reverse of $s$ in which all occurrences of \bcar\ are replaced by \acar\ and \bcdr\ replaced by \acdr. Clearly, $\simpl{\{s\overline{s}\}} = \epsilonset$.
 
We next prove the completeness and minimality of \compautomaton{A}{\pi}.
\begin{lemma}
  \label{thm:aecomp}
$\{s\mid\Simpl(\lang{\automaton{M}{\pi}{\{s\}}}) \neq\ \emptyset\}
  = \concatenate{\lang{\compautomaton{A}{\pi}}}{(\acar + \acdr)^\ast} $
\begin{proof}
  We first prove  $LHS\ \subseteq\ RHS$.  Let  the string $s
  \in   \Simpl(\lang{\automaton{M}{\pi}{\{s\}}})$.   Then   by
  Lemma~\ref{thm:grammar-eq},             $s             \in
  \Simpl(\concatenate{\lang{\automaton{M}{\pi}{\{\epsilon\}}}}{\{s\}})$.
  By Property~\ref{prop:canon-2},  this also means  that $s
   \in
  \Simpl(\concatenate{\canonical{\lang{\automaton{M}{\pi}{\{\epsilon\}}}}}{\{s\}})$. Since
  strings                                                 in
  $\canonical{\lang{\automaton{M}{\pi}{\{\epsilon\}}}}$ are of
  the   form  $(\acar   +  \acdr   +  \clazy)^\ast(\bcar   +
  \bcdr))^\ast$  (Property~\ref{prop:canon-1}),  this  means
  that there is a  string $\concatenate{p_1}{p_2}$ such that
  $p_1 \in  (\acar + \acdr + \clazy)^\ast$  and $p_2 
  \in  (\bcar  +  \bcdr)^\ast$,  and  $\Simpl(\{p_2\}\{s\})\
  \subseteq  (\acar +  \acdr)^\ast$. Thus  $s$ can  be split
  into   two   strings   $s_1$    and   $s_2$,   such   that
  $\Simpl(\concatenate{\{p_2\}}{\{s_1\}})   =  \epsilonset$.
  Therefore $s_1  = \overline{p_2}$.  From  the construction
  of  \compautomaton{A}{\pi}   we  have   $\overline{p_2}  \in
  \lang{\compautomaton{A}{\pi}}$  and  $s_2\   \in\  (\acar  +
  \acdr)^\ast$.            Thus,           $s\           \in
  \concatenate{\lang{\compautomaton{A}{\pi}}}{(\acar         +
    \acdr)^\ast}$.

  Conversely,  for the  proof  of $RHS\  \subseteq\ LHS$,  we
  assume       that       a        string       $s       \in
  \concatenate{\lang{\compautomaton{A}{\pi}}}{(\acar         +
    \acdr)^\ast}$.      From     the     construction     of
  \compautomaton{A}{\pi} we  have strings $p_1, p_2,  s'$ such
  that                      $p_1p_2                      \in
  \canonical{\lang{\automaton{M}{\pi}{\epsilon}}}$,  $p_1  \in
  (\acar  + \acdr  +  \clazy)^\ast$,  $p_2  \in (\bcar  +
  \bcdr)^\ast$, $s$ is  $\overline{p_2}s'$ and $s'\in (\acar
  +                   \acdr)^\ast$.                    Thus,
  $\Simpl(\lang{\automaton{M}{\pi}{\{s\}}})                  =
  \Simpl(\lang{\concatenate{\automaton{M}{\pi}{\epsilonset}}{\{s\}}})
  =\Simpl(\concatenate{\canonical{\lang{\automaton{M}{\pi}{\{\epsilon\}}}}{\{s\}}}) =
  \Simpl(\{\concatenate{p_1}{\concatenate{p_2}{\concatenate{\overline{p_2}}{s'}}}\})
  =                    \{p_1s'\}$.                     Thus,
  $\Simpl(\lang{\automaton{M}{\pi}{\{s\}}})$ is  non-empty and
  $s \in LHS$.
\end{proof}

\end{lemma}

We   now    prove   our    main   result:   Our    slicing   algorithm
represented by \KEEP~(Algorithm~\ref{algo:comp-automata}) returns true if and only
if   \Simpl(\concatenate{\lang{\automaton{A}{\pi}{\epsilon}}}{\sigma})
is non-empty.

\begin{theorem}
  $\Simpl(\lang{\automaton{M}{\pi}{\sigma}}) \neq \emptyset
  \leftrightarrow\ \KEEP(e,\sigma)$
\end{theorem}
\begin{proof}
  We    first    prove     the    forward    implication.     Let
  $s    \in    \Simpl(\lang{\automaton{M}{\pi}{\sigma}})$.     From
  Lemma~\ref{thm:grammar-eq},
  $s                                                          \in
  \Simpl(\concatenate{\lang{\automaton{M}{\pi}{\epsilon}}}{\sigma})$.
  From                               Property~\ref{prop:canon-2},
  $s                    \in                   \Simpl(\concatenate
  {\canonical{\lang{\automaton{M}{\pi}{\epsilon}}}}{\sigma})$.
  Thus,    there   are    strings    $p_1,    p_2$   such    that
  $p_1  \in  \canonical{\lang{\automaton{M}{\pi}{\epsilon}}},\  p_2
  \in        \sigma,        s       =        \Simpl(\{p_1p_2\})$.
  Further $p_1$ in  turn can be decomposed as  $p_3p_4$ such that
  $p_3    \in    (\acar    +   \acdr    +    \clazy)^\ast$    and
  $p_4    \in   (\bcar    +   \bcdr)^\ast$.     We   also    have
  $\Simpl(\{p_4p_2\})  \subseteq  (\acar  +  \acdr)^\ast$.   Thus
  $\overline{p_4}$ is a prefix of $p_2$.

 From   the  construction   of   $\compautomaton{A}{\pi}$,  we   know
 $\overline{p_4} \in
 \lang{\compautomaton{A}{\pi}}$.   Further,   $\overline{p_4}$ is a prefix of  $p_2$ and $p_2 \in \sigma$, from
 the    prefix   closed    property   of    $\sigma$   we    have
 $\overline{p_4}  \in \sigma$.  This  implies
 $\compautomaton{A}{\pi} \cap \sigma \neq \emptyset$  and  thus
 $\KEEP(e, \sigma)$ returns true.

 Conversely, if $\KEEP(e, \sigma)$ is true, then $\exists
s:s \in \lang{\compautomaton{A}{\pi}} \cap\sigma$. In particular, 
$s \in \lang{\compautomaton{A}{\pi}}$. Thus, from
Lemma~\ref{thm:aecomp} we have
$\Simpl(\lang{\automaton{M}{\pi}{\{s\}}}) \neq \emptyset$. Further, since $s
\in \sigma$ we have $\Simpl(\lang{\automaton{M}{\pi}{\sigma}}) \neq \emptyset$.\end{proof}


\section{Extension  to higher order functions}
\label{sec:higher-order}
\begin{figure}[t!]
  \scalebox{0.89}{\renewcommand{\arraystretch}{1}
  \begin{minipage}{.39\textwidth}
  \begin{tabular}[b]{|@{\ }l@{\ }|}\hline 
    \begin{uprogram}
      \UFL(\DEFINE\ (\hof~\pf~\pl) 
      \UNL{1}  (\SRETURN\ $\pi$:(\pf\ \pl)))
    \end{uprogram}\\ \\ \\ \\ \\
    \begin{uprogram}
      \UFL (\DEFINE\ (\foldr~\f~\id\ \pl)
      \UNL{1} (\SIF (\NULLQ\ \pl)) (\SRETURN\ \id)
      \UNL{2} (\SRETURN\ (\f\ (\CAR\ \pl)\
      \UNL{7} \ (\foldr~\f~\id\ (\CDR~\pl)))))
    \end{uprogram}\\ \\
    \begin{uprogram}
      \UFL(\DEFINE\ (\fun~\px~\py) 
      \UNL{1}  (\SRETURN\ (+\ \py\ 1)))
    \end{uprogram} \\ \\ \\
    \begin{uprogram}
      \UFL (\DEFINE\ (\mainpgm)
      \UNL{1} $(\LET\  \plone\  \leftarrow$ (\CONS~\pa~(\CONS ~\pb~\NIL)) \IN
      \UNL{2}  $(\LET\  \g\  \leftarrow$  (\foldr~\fun~0) \IN
      \UNL{3} $(\SRETURN\ (\CONS\ (\hof~\CAR~\plone)$
      \UNL{3} $\phantom{(\SRETURN\ (\CONS\ }(\hof~\g~\plone)))))$
    \end{uprogram} \\ \hline \\
    (a) A program with higher order functions\\ \hline
  \end{tabular}
  \end{minipage}
  \begin{minipage}{.355\textwidth}
  \begin{tabular}[b]{|@{\ }l@{\ }|}\hline
    \begin{uprogram}
      \UFL (\DEFINE\ (\hofeven~l)
      \UNL{1}  (\SRETURN\ $\pi_{f}$:(\foldrf~0~l)))
    \end{uprogram}\\ \\
    \begin{uprogram}
      \UFL (\DEFINE\ (\hofcar~\pl)
      \UNL{1}  (\SRETURN\ $\pi_{c}$:(\CAR\ \pl)))
    \end{uprogram}\\ \\
    \begin{uprogram}
      \UFL (\DEFINE\ (\foldrf~\id\ \pl)
      \UNL{1} (\SIF (\NULLQ\ \pl)) (\SRETURN\ \id)
      \UNL{2} (\SRETURN\ (\fun\ (\CAR\ \pl)\
      \UNL{7} \ (\foldrf~\id\ (\CDR~\pl)))))
    \end{uprogram} \\ \\
    \begin{uprogram}
      \UFL(\DEFINE\ (\fun~\px~\py) 
      \UNL{1}  (\SRETURN\ (+\ \py\ 1)))
      \UNL{1}
    \end{uprogram}\\  \\
    \begin{uprogram}
      \UFL (\DEFINE\ (\mainpgm)
      \UNL{1} $(\LET\  \plone\  \leftarrow$ (\CONS~\pa~(\CONS ~\pb~\NIL)) \IN
      \UNL{2}  $(\LET\  \g\  \leftarrow$  (\foldrf~0) \IN
      \UNL{3} $(\SRETURN\ (\CONS\ (\hofcar~\plone)$
      \UNL{3} $\phantom{(\SRETURN\ (\CONS\ }(\hofeven~\plone))))))$
    \end{uprogram} 
    \\ \hline \\
      (b) Program in (a) after specialization.\\ \hline
  \end{tabular}
  \end{minipage}
  \begin{minipage}{.27\textwidth}
  \begin{tabular}[b]{|@{\ }p{49mm}@{\ }|}\hline
  \begin{uprogram}
    \UFL(\DEFINE\ (\hof~\pf~\pl) 
    \UNL{1}  (\SRETURN\ $\pi$:(\pf\ \pl)))
    \\ \\ \\ \\ \\
    \\ \\ \\ \\ \\
    \\ \\ \\ \\
    $(\DEFINE\ (\mainpgm)$
    \UNL{1}  $(\LET\  \plone\  \leftarrow$ (\CONS~\pa~\removed)       \IN
    \UNL{2}  $(\LET\  \g\  \leftarrow$  \removed) \IN
    \UNL{3} $(\SRETURN\ (\CONS\ (\hof~\CAR~\plone)$
    \UNL{3} $\phantom{(\SRETURN\ (\CONS\ } \removed))))$
    \end{uprogram} \\ \hline
      (c) Slice of the program in (a) with \mbox{~\quad}slicing
      criterion $\{\epsilon, \acar\}$. \\ \hline
  \end{tabular}
  \end{minipage}}
  \caption{An example higher order program}\label{fig:higher-order-mot-example}
\end{figure}

We now describe how our method can  also be used to slice higher order
programs.  This section has been included mainly for completeness, and
we do not make claims of  novelty. We handle all forms of higher-order
functions except  the cases of  functions being returned as  a result,
and functions  being stored  in data  structures---in our  case lists.
Even with  these limitations,  one can  write a  number of  useful and
interesting higher-order programs in our language. 

  Consider the  program in Figure~\ref{fig:higher-order-mot-example}(a).
  It contains a  higher order function $\hof$ which  applies its first
  argument  $\pf$   on  its  second  argument   $\pl$.   The  function
  $\mainpgm$  creates  a list  \plone\/  and  a function  value  \pg\/
  (through partial  application) and  uses these in  the two  calls to
  $\hof$.  Finally, \mainpgm\  returns the result of these  calls in a
  pair.  The program  exhibits  higher order  functions  that take  as
  actual arguments both manifest functions and partial applications.
 
  For our  first-order method  to work on  higher order  functions, we
  borrow               from                a                technique
  called                    \emph{firstification}~\cite{Mitchell:2009,
  Reynolds98a}. Firstification transforms a  higher-order program to a
  first-order program  without altering its semantics.  Our version of
  firstification repeatedly  (i) finds for each  higher-order function
  the bindings of each of its functional parameters, (ii) replaces the
  function  by a  specialized version  for each  of the  bindings, and
  (iii) replaces  each application  of f  by its  specialized version.
  These  steps  are repeated  till  we  we  are  left with  a  program
  containing  first  order  functions  only.   In  the  example  being
  considered, we  first discover  that {\tt  f} in  {\tt foldr}  has a
  single binding  to {\tt  fun} and  the {\tt  f} of  {\tt hof}  has a
  binding  to  {\tt car}.   Specialization  gives  the functions  {\tt
  foldr\_fun} and {\tt hof\_car}. We now see that {\tt f} of {\tt hof}
  has a second  binding to the partial application  {\tt (foldr fun)},
  This gives rise to a second  specialization of {\tt hof} called {\tt
  hof\_g}.

 The     program     after      firstification     is     shown     in
 Figure~\ref{fig:higher-order-mot-example}(b).    This    program   is
 subjected to demand analysis and  the results are reflected back into
 the higher-order  program.  Inside a  higher order function  that has
 been specialized,  the demand  on an  expression is  an union  of the
 demands on  the specialized  versions of  the expression.   Thus, the
 demand on $\pi$ is given by the union of the demands on $\pi_{c}$ and
 $\pi_{f}$. Where the higher order  function is applied, the demand on
 its  arguments  is  derived  from   the  demand  transformer  of  its
 specialized version. As an example, the  demand on {\tt lst1} in {\tt
 (hof  car lst1)}  is obtained  from the  demand transformers  of {\tt
 hof\_car}.  For the slicing  criterion $\{\epsilon, \acar\}$, the the
 demand on  the second argument  of {\tt (cons  (hof car lst1)  (hof g
 lst1))} is null and thus this argument and the binding of {\tt g} can
 both be sliced away. The slice  for $\{\epsilon, \acar\}$ is shown in
 Figure~\ref{fig:higher-order-mot-example}(c).

Note that  our simple  firstifier requires us  to statically  find all
bindings of a  functional parameter. This is not possible  if we allow
functions  to   be  returned   as  results   or  store   functions  in
data-structures. As an  example we can consider a  function $f$, that,
depending on a  calculated value $n$, returns a  function $g$ iterated
$n$             times            (i.e.             $g            \circ
g\;\circ \stackrel{n~\mathit{times}}{\ldots}  \circ\; g$).   A higher-order
function receiving this  value as a parameter  would be cannot be specialized using the techniques described, for example, in~\cite{Mitchell:2009}.  A similar
thing can happen if we allow functions in lists.

\begin{table}[t]
\caption{Statistics for incremental and non-incremental slicing}
\label{tab:exp-results}
\centering

\renewcommand{\arraystretch}{.89}
\begin{tabular}{@{}|@{\ }l|r|@{}r@{\ }|r|r|r|r|r|r|r|r|r|}\hline
  Program &
  \multicolumn{1}{l|}{Pre-} &
  \#exprs&
  \multicolumn{3}{|c|}{Slicing with $\{\epsilon\}$} &
  \multicolumn{3}{|c|}{Slicing with $\{\epsilon, 0\}$} &
  \multicolumn{3}{|c|}{Slicing with $\{\epsilon, 1\}$}\\ \cline{4-12}
  &
  \multicolumn{1}{@{\ }p{1.0cm}@{\ }|}{comput ation}&
  \multicolumn{1}{@{\ }p{1.0cm}@{\ }|}{in program} &
  \multicolumn{1}{|@{\ }p{1.0cm}@{\ }|}{Non-inc time (ms)} &
  \multicolumn{1}{|@{\ }p{1.0cm}@{\ }|}{Inc time (ms)} &
  \multicolumn{1}{|@{\ }p{0.9cm}@{\ }|}{\#expr in slice} &
  \multicolumn{1}{|@{\ }p{1.0cm}@{\ }|}{Non-inc time (ms)} &
  \multicolumn{1}{|@{\ }p{1.0cm}@{\ }|}{Inc time (ms)} &
  \multicolumn{1}{|@{\ }p{0.9cm}@{\ }|}{\#expr in slice} &
  \multicolumn{1}{|@{\ }p{1.0cm}@{\ }|}{Non-inc time (ms)} &
  \multicolumn{1}{|@{\ }p{1.0cm}@{\ }|}{Inc time (ms)} &
  \multicolumn{1}{|@{\ }p{0.9cm}@{\ }|}{\#expr in slice}\\
  \hline
  \multicolumn{12}{|c|}{First-order Programs} \\
\hline
treejoin &6900.0&581&6163.2&2.4&536&5577.2&2.8&538&5861.4&4.6&538 \\
\hline
deriv &399.6&389&268.0&1.6&241&311.2&1.6&249&333.2&2.3&266 \\
\hline
paraffins &3252.8&1152&2287.3&5.2&1067&2529.2&5.1&1067&2658.7&5.1&1067 \\
\hline
nqueens &395.4&350&309.9&1.5&350&324.6&1.5&350&328.1&1.6&350 \\
\hline
minmaxpos &27.9&182&18.1&0.9&147&19.5&0.8&149&20.5&0.9&149 \\
\hline
nperm &943.1&590&627.4&2.1&206&698.4&11.2&381&664.0&11.8&242 \\
\hline
linecharcount\!\! &11.7&91&7.0&0.5&69&7.5&0.5&78&7.4&0.5&82 \\
\hline
studentinfo &1120.6&305&858.2&1.2&96&854.6&1.3&101&1043.3&7.5&98 \\
\hline
knightstour &2926.5&630&2188.1&2.8&436&2580.6&12.2&436&2492.8&7.4&436 \\
\hline
takl &71.6&151&46.1&0.7&99&49.5&0.8&105&48.5&0.7&99 \\
\hline
lambda &4012.9&721&3089.0&2.7&26&3377.4&13.2&705&2719.8&5.3&33 \\
\hline
\multicolumn{12}{|c|}{Higher-order Programs} \\ \hline
parser &60088.2&820&46066.8&2.3&203&45599.0&2.3&209&61929.2&4.1&209 \\
\hline
maptail &22.1&96&5.5&0.5&51&15.4&0.6&67&17.4&0.6&56 \\
\hline
fold &21.4&114&13.3&0.4&17&14.4&0.5&76&16.9&0.6&33 \\


\hline
\end{tabular}

\end{table}


\section{Experiments and results}\label{sec:exp-result}
In this  section, we present the  results from our experiments  on the
implementations of  both versions of  slicing.  In the absence  of the
details of implementations of other  slicing methods, we have compared
the  incremental   step  of   our  method  with   the  non-incremental
version.
Our experiments show that the incremental slicing algorithm gives
benefits  even  when  the  overhead of  creating  the  completing
automata is amortized  over  even a 
few slicing criteria.

Our benchmarks consists of first order programs derived from the nofib
suite~\cite{nofib}. The higher order programs have been handcrafted to
bring out the issues related  to higher order slicing.  The program named 
{\bf parser} includes most of the higher  order parser combinators
required for  parsing. {\bf fold} corresponds  to the
example    in     Figure~\ref{fig:higher-order-mot-example}.     
Table~\ref{tab:exp-results} shows  the time required for  slicing with
different slicing  criteria.  For each  benchmark, we first  show, the
pre-computation time,  i.e.  the  time required to construct the
completing  automata.  We  then   consider  three  different  slicing
criteria, and  for each  slicing criterion, present the  times for
non-incremental slicing and the incremental step.
The  results   in  Table~\ref{tab:exp-results}   show  that   for  all
benchmarks, the  time required to  compute the completing  automata is
comparable   to    the   time   taken   for    computing   the   slice
non-incrementally.  Since computing completing  automata is a one time
activity, incremental slicing is very efficient even when a program is
sliced only  twice.  As seen in  Table~\ref{tab:exp-results}, the time
taken  for the  incremental step  is orders  of magnitude  faster than
non-incremental slicing,  thus confirming the benefits  of reusing the
completing automata.  

We also show the number of  expressions in the original program and in
the slice  produced to  demonstrate the  effectiveness of  the slicing
process itself.  Here are some  of the  interesting cases.  It  can be
seen  that the  slice  for  {\bf nqueens}  for  any slicing  criterion
includes the  entire program.  This  is because finding out  whether a
solution exists  for {\bf nqueens}  requires the entire program  to be
executed.   On  the  other  hand,   the  program  {\bf  lambda}  is  a
$\lambda$-expression evaluator  that returns a tuple  consisting of an
atomic value  and a  list.  The  criterion $\{\epsilon,  0\}$ requires
majority of the expressions in the  program to be present in the slice
to  compute  the atomic  value.   On  the  other hand,  the  criterion
$\{\epsilon\}$ or  $\{\epsilon, 1\}$  do not require  any value  to be
computed and expressions  which compute the constructor  only are kept
in  the slice,  hence our  algorithm is  able to  discard most  of the
expressions.  This  behavior can be  clearly seen in  the higher-order
example  {\bf fold}  where a  slicing criterion  $\{\epsilon, \acar\}$
selects an expression which only uses the first element of {\tt lst1},
thus allowing our slicing algorithm to discard most of the expressions
that construct  {\tt lst1}.
After  examining the  nature of  the benchmark  programs, the  slicing
criteria and  the slices, we  conclude that slicing is  most effective
when the slicing criterion selects  parts of a bounded structure, such
as a tuple, and  the components of the tuple are  produced by parts of
the program that are largely disjoint.

\section{Related work} 
\label{sec:relatedwork}
  Program slicing has been an  active area of research.  However,
  most  of  the  efforts  in slicing  have  been  for  imperative
  programs.     The    surveys~\cite{Tip95asurvey,    BinkleyH04,
    Silva:2012} give  good overviews  of the variants  of the
  slicing problem and their  solution techniques.  The discussion
  in this  section will   be centered mainly around   static and
  backward  slicing  of  functional programs.

  In the context of imperative programs, a slicing criterion is a
  pair consisting  of a  program point, and  a set  of variables.
  The slicing problem is to  determine those parts of the program
  that  decide  the  values  of  the  variables  at  the  program
  point~\cite{weiser84}.   A  natural  solution  to  the  slicing
  problem  is through  the use  of data  and control  dependences
  between  statements.    Thus  the  program  to   be  sliced  is
  transformed into  a graph  called the program  dependence graph
  (PDG)~\cite{Ottenstein:1984:PDG,  horwitz88},  in  which  nodes
  represent individual statements and edges represent dependences
  between them.  The slice consists of  the nodes in the PDG that
  are reachable  through a  backward traversal starting  from the
  node  representing the  slicing criterion.   Horwitz, Reps  and
  Binkley~\cite{horwitz88} extend PDGs  to handle interprocedural
  slicing.   They  show that  a  naive  extension could  lead  to
  imprecision in the computed slice due to the incorrect tracking
  of  the calling  context.   Their solution  is  to construct  a
  context-independent summary of each  function through a linkage
  grammar,  and then  use this  summary to  step across  function
  calls. The resulting graph is  called a system dependence graph
  (SDG).  Our method generalizes  SDGs to additionally keep track
  of the  construction of algebraic data  types (\CONS), selection
  of  components  of  data   types  (\CAR\/  and  \CDR)  and  their
  interaction, which may span across functions.
 
  Silva,  Tamarit  and  Tom\'as  ~\cite{Silva_System_Dependence_Graph}
  adapt  SDGs for  functional  languages, in  particular Erlang.   The
  adaptation is  straightforward except  that they  handle dependences
  that arise  out of pattern  matching.  Because  of the use  of SDGs,
  they can manage calling contexts precisely.  However, as pointed out
  by    the    authors    themselves,   when    given    the    Erlang
  program:                                           \{\mainpgm()~{\tt
  ->}~\px~=~\{1,2\},~\{\py,\pz\}~=~\px,~\py\},  their method  produces
  the         imprecise        slice         \{\mainpgm()~        {\tt
  ->}~ \px~=~\{1,2\},  \{\py,\removed \}~=~\px,  \py\} when  sliced on
  the variable \py. Notice that the  slice retains the constant 2, and
  this  is   because  of   inadequate  handling  of   the  interaction
  between \CONS\/  and \CDR.  For  the equivalent program (\LET  ~ \px
  $\leftarrow$  ~  (\CONS~1~2)~ \IN  ~  (\LET  ~ \py  ~  $\leftarrow$~
  (\CAR~ \px)~ \IN ~ \py))  with the slicing criterion $\epsilon$, our
  method would correctly compute the demand on the constant 2 as
  $\bcdr  (\epsilon  \cup  \acar)$.  This  simplifies  to  the  demand
  $\emptyset$, and 2 would thus not  be in the slice. Another issue is
  that  while the  paper  mentions  the need  to  handle higher  order
  functions,  it  does  not  provide details  regarding  how  this  is
  actually done.   This would  have been interesting  considering that
  the language considered allows lambda expressions.

  The slicing  technique that is closest  to ours is due  to Reps
  and  Turnidge~\cite{reps96}.   They use  projection  functions,
  represented  as  certain kinds  of  tree  grammars, as  slicing
  criteria.  This  is the same as  our use  of prefix-closed
  regular expressions.  Given a program \emph{P} and a projection
  function  $\psi$, their  goal  is to  produce  a program  which
  behaves  like $\psi\circ\emph{P}$.   The  analysis consists  of
  propagating   the   projection   function  backwards   to   all
  subexpressions   of  the   program.   After   propagation,  any
  expression with  the projection function  $\bot$ (corresponding
  to our $\emptyset$ demand), are sliced out of the program.  Liu
  and  Stoller~\cite{Liu:2003} also  use  a method  that is  very
  similar to~\cite{reps96}, but more extensive in scope.

  These  techniques  differ from  ours  in  two respects.   These
  methods,  unlike   ours,  do  not   derive  context-independent
  summaries of functions.  This results  in a loss of information
  due to  merging of  contexts and affects  the precision  of the
  slice.  Moreover,  the computation of function  summaries using
  symbolic demands enables the incremental version of our slicing
  method.   Consider,   as  an  example,  the   program  fragment
  $\pi\!\!:\!(\CONS~\pi_1\!\!:\!x~\pi_2\!\!:\!y)$    representing
  the  body of  a function.   Demand analysis  with the  symbolic
  demand   $\sigma$  gives   the   demand  environment   $\left\{
    \pi\mapsto          \sigma,          \pi_1\mapsto\bcar\sigma,
    \pi_2\mapsto\bcdr\sigma \right\}$.   Notice that  the demands
  $\pi_1$ and  $\pi_2$ are  in terms of  the symbols  \bcar\/ and
  \bcdr. This is  a result of our decision to  work with symbolic
  demands,   and,    as   a   consequence,   also    handle   the
  constructor-selector interaction symbolically.  If we now slice
  with  the default  criterion $\epsilon$  and then  canonicalize
  (instead of simplify), we are  left with the demand environment
  $\left\{      \pi\mapsto      \epsilon,      \pi_1\mapsto\bcar,
    \pi_2\mapsto\bcdr  \right\}$.  Notice  that  there is  enough
  information in  the demand  environment to deduce,  through the
  construction   of  the   completing  automaton,   that  $\pi_1$
  ($\pi_2$) will  be in the  slice only if the  slicing criterion
  includes  \acar(\acdr).  Since  the methods  in~\cite{reps96} and
  \cite{Liu:2003} deal with demands in their concrete forms,   it
  is  difficult  to see the incremental version being  replayed
  with their methods. 

  There are  other less related  approaches to slicing.   A graph
  based   approach  has   also   been  used   by  Rodrigues   and
  Barbosa~\cite{Rodrigues05componentidentification} for component
  identification in Haskell programs. Given the intended use, the
  nodes  of  the  graph  represents coarser  structures  such  as
  modules,  functions and  data type  definitions, and  the edges
  represents  relations  such  as  containment  (e.g.   a  module
  containing a  function definition).  On a  completely different
  note,  Rodrigues   and  Barbosa~\cite{Rodrigues_jucs_12_7}  use
  program  calculation   in  the  Bird-Meerteens   formalism  for
  obtaining  a  slice.  Given  a  program  $P$ and  a  projection
  function $\psi$,  they calculate a program  which is equivalent
  to $\psi\circ\emph{P}$.   However the method is  not automated.
  Finally,  dynamic  slicing  techniques have  been  explored  for
  functional programs by Perera  et al.~\cite{perera12}, Ochoa et
  al.~\cite{Ochoa:2008}     and Biswas~\cite{Biswas:1997}.

\section{Conclusions and Future Work}\label{sec:concl}
We  have  presented  a  demand-based  algorithm  for  incremental
slicing  of  functional programs.   The  slicing  criterion is  a
prefix-closed regular language and represents parts of the output
of the  program that  may be  of interest to  a user of our slicing
method.   We view  the slicing  criterion  as a  demand, and  the
non-incremental version of  the slicer does a  demand analysis to
propagate this demand through the program.  The slice consists of
parts of the program with non-empty demands after the propagation.  A
key  idea in  this analysis  is the  use of  symbolic demands  in
demand  analysis. Apart form better handling of calling
contexts that improves the precision of the analysis, 
this also helps in building the incremental version.

The incremental version builds on the non-incremental version.  A
per  program pre-computation  step  slices the  program with  the
default  criterion   $\epsilon$.   This  step  factors   out  the
computation that  is common  to slicing  with any  criterion. The
result, reduced to a canonical form,  can now be used to find the
slice for  a given criterion  with minimal computation.   We have
proven the correctness of  the incremental algorithm with respect
to the  non-incremental version.   And finally, we  have extended
our     approach     to     higher-order     programs     through
firstification. Experiments with  our implementation confirm the 
benefits of incremental slicing.
    
There are however two areas of concern, one related to efficiency
and the  other to precision. To  be useful, the slicer  should be
able  to slice  large  programs quickly.   While our  incremental
slicer  is  fast  enough,   the  pre-computation  step  is  slow,
primarily  because  of  the canonicalization step.   In
addition, the firstification process may create a large number of
specialized first-order programs. As  an example, our experiments
with  functional  parsers  show   that  the  higher-order  parser
combinators such  as or-parser  and the  and-parser are called 
often, and the arguments to these calls are in turns calls to 
higher order functions, for instance the Kleene closure and the
positive closure parsers. 

\begin{figure}[t!]
       \begin{uprogram}
      \UFL(\DEFINE\ (\mapsq~\pl) 
      \UNL{1} (\SIF (\NULLQ\ \pl) (\SRETURN\ \pl)
      \UNL{2} (\SRETURN\ (\CONS\ (\sq\ (\CAR\ \pl))
      \UNL{7} \ (\mapsq~(\CDR~\pl)))))
       \end{uprogram}
       \caption{Example to illustrate the imprecision due to
         Mohri-Nederhoff approximation}
       \label{fig:imprecision}
 \end{figure}

The other concern is that  while our polyvariant approach through
computation  of   function  summaries  improves   precision,  the
resulting  analysis   leads  to  an  undecidable   problem.   The
workaround  involves   an  approximation   that  could   lead  to
imprecision. As an example, consider the function
$\mathbf{mapsq}$ shown in Figure~\ref{fig:imprecision}.
 The reader  can verify  that the  function summary  for $\mapsq$
 would     be     given     as:     $\Lf{\mapsq}{1}{\sigma}     =
 \concatenate{\Lfun{\mapsq}{1}}{\sigma}$, where   $\Lfun{\mapsq}{1} $  is
 the language    $\epsilon    \mid    \acdr^{n}\bcdr^{n}    \mid
 \acdr^{n}\acar\clazy\bcar\bcdr^{n}$, for $n\geq 0$.   Now,   given  a   slicing
 criterion    $\sigma    =   \{\epsilon,    \acdr,    \acdr\acdr,
 \acdr\acdr\acar \}$ standing for the path to the third element of a list, it
 is   easy   to   see  that   $\Lf{\mapsq}{1}{\sigma}   $   after
 simplification would  give back $\sigma$ itself,  and this is  the most
 precise  slice. However,  due  to Mohri-Nederhoff  approximation
 $\Lfun{\mapsq}{1}$  would  be  approximated  by  $\epsilon  \mid
 \acdr^{n}\bcdr^{m} \mid  \acdr^{k}\acar\clazy\bcar\bcdr^{l}$,
 $n$, $m$, $k$, $l$ $\geq 0$. In
 this  case, $\Lfun{\mapsq}{1}$  would  be  $(\acar +  \acdr)^*$,
 keeping all the elements of the input list $\pl$ in the slice.

\balance
\bibliography{fun_slicing_bib}{}

\end{document}